\documentclass[11pt]{article}
\usepackage[utf8x]{inputenc}

\usepackage{a4}
\usepackage{verbatim}
\usepackage{textcomp}
\usepackage{latexsym}
\usepackage{amsmath}
\usepackage{amsfonts}
\usepackage{amssymb}
\usepackage{amsthm}
\usepackage{hyperref}
 
\newtheorem{theorem}{Theorem}
\newtheorem{theorem*}{Theorem}
\newtheorem{corollary}[theorem]{Corollary}
\newtheorem{lemma}[theorem]{Lemma}

\newtheorem{proposition}[theorem]{Proposition}
\newtheorem{definition}[theorem]{Definition}
\newtheorem{claim}[theorem]{Claim}

\newtheorem{remark}[theorem]{Remark}

\newcommand{\tU}{\mathcal{UT}}

\newcommand{\F}{\mathbb{F}}

\newcommand{\poly}{\mathrm{poly}}

\newcommand{\PER}{\mathrm{PER}}

\newcommand{\DET}{\mathrm{DET}}

\title{Some Lower Bound Results for Set-Multilinear Arithmetic Computations}
\author
{
V. Arvind \quad S. Raja \\\\ Institute of Mathematical Sciences, Chennai. \\ \{arvind,rajas\}@imsc.res.in
}

\begin{document}

\maketitle

\begin{abstract}
In this paper, we study the structure of set-multilinear arithmetic
circuits and set-multilinear branching programs with the aim of
showing lower bound results. We define some natural restrictions of
these models for which we are able to show lower bound results.  Some
of our results extend existing lower bounds, while others are new and
raise open questions. More specifically, our main results are the
following:

\begin{itemize}
\item We observe that set-multilinear arithmetic circuits can be
  transformed into shallow set-multilinear circuits efficiently,
  similar to depth reduction results of \cite{VSBR83,RY08} for more
  general commutative circuits. As a consequence, we note that
  polynomial size set-multilinear circuits have quasi-polynomial size
  set-multilinear branching programs. We show that \emph{narrow}
  set-multilinear ABPs (with a restricted number of set types)
  computing the Permanent polynomial $\PER_n$ require
  $2^{n^{\Omega(1)}}$ size. A similar result for general
  set-multilinear ABPs appears difficult as it would imply that the
  Permanent requires superpolynomial size set-multilinear circuits. It
  would also imply that the noncommutative Permanent requires
  superpolynomial size noncommutative arithmetic circuits.

\item Indeed, we also show that set-multilinear branching programs are
  exponentially more powerful than \emph{interval} multilinear
  circuits (where the index sets for each gate is restricted to be an
  interval w.r.t.\ some ordering), assuming the sum-of-squares
  conjecture. This further underlines the power of set-multilinear
  branching programs.

\item Finally, we consider set-multilinear circuits with restrictions
  on the number of proof trees of monomials computed by it, and prove
  exponential lower bounds results. This raises some new lower bound
  questions.

\end{itemize}

\end{abstract}

\section{Introduction}\label{one}

Let $\F$ be a field and $X=X_1\sqcup X_2\sqcup\dots\sqcup X_d$ be a
partition of the variable set $X$. A \emph{set-multilinear polynomial}
$f\in\F[X]$ w.r.t.\ this partition is a homogeneous degree $d$
multilinear polynomial such that every nonzero monomial of $f$ has
exactly one variable from $X_i, 1\le i\le d$.

Both the Permanent polynomial $\PER_n$ and the Determinant polynomial
$\DET_n$ are set-multilinear polynomials. The variable set is
$X=\{x_{ij}\}_{1\le i,j\le n}$ and the partition can be taken as the
row-wise partition of the variable set. I.e.\ $X_i=\{x_{ij}\mid 1\le
j\le n\}$ for $1\le i\le n$.

In this paper we will consider set-multilinear circuits and
set-multilinear branching programs for computing set-multilinear
polynomials. Set-multilinear circuits are well studied. The model of
set-multilinear branching programs that we consider is more general
than related notions of branching programs recently studied in the
literature, like the read-once oblivious branching programs (ROABPs) \cite{FS13}.

A \emph{set-multilinear arithmetic circuit} $C$ computing $f$
w.r.t.\ the above partition of $X$, is a directed acyclic graph such
that each in-degree $0$ node of the graph is labelled with an element
from $X\cup \F$. Each internal node $v$ of $C$ is either a $+$ gate or
$\times$ gate. With each gate $v$ of we can associate a subset of
indices $I_v\subseteq [d]$ and the polynomial $C_v$ computed by the
circuit at $v$ is set-multilinear over the variable partition
$\sqcup_{i\in I_v} X_i$. If $v$ is a $+$ gate then for each input $u$
of $v$ $I_u=I_v$, and $v$ is a $\times$ gate with inputs $v_1$ and
$v_2$ then $I_v= I_{v_1}\sqcup I_{v_2}$. Clearly, in a set-multilinear
circuit every gate computes a set-multilinear polynomial (in a
syntactic sense). The output gate is labeled by $[d]$ and computes the
polynomial $f$.

A \emph{set-multilinear algebraic branching program} smABP is a
layered directed acyclic graph (DAG) with one in-degree zero vertex
$s$ and one out-degree zero vertex $t$. The vertices of the graph are
partitioned into layers $0,1,\ldots,d$, and edges go only from layer
$i$ to $i + 1$ for each $i$. The source is the only vertex at level
$0$ and the sink is the only vertex at level $d$. We can associate an
index set $I_v\subseteq [d]$ with each node $v$ in the smABP, and the
polynomial computed at $v$ is set-multilinear w.r.t.\ the partition
$\sqcup_{i\in I_v} X_i$. For any edge $(u,v)$ in the branching program
labeled by a homogeneous linear form $\ell$, we have $I_v = I_u\sqcup
\{i\}$ for some $i\in [d]$, and $\ell$ is a linear form over variables
  $X_i$. The size of the ABP is the number of vertices.

For any $s$-to-$t$ directed path $\gamma=e_1,e_2,\ldots,e_d$, where
$e_i$ is the edge from level $i-1$ to level $i$, let $\ell_i$ denote
the linear form labeling edge $e_i$. Let
$f_\gamma=\ell_1\cdot\ell_2\cdots\ell_d$ be the product of the linear
forms in that order. Then the ABP computes the set-multilinear degree
$d$ polynomial:
\[
f = \sum_{\gamma\in\mathcal{P}} f_\gamma,
\]
where $\mathcal{P}$ is the set of all directed paths from $s$ to $t$.

\begin{remark}
Showing a superpolynomial lower bound for set-multilinear circuits and
even for set-multilinear ABPs for computing the Permanent polynomial
is an open problem. In this paper we discuss some restricted versions
of set-multilinear branching programs and show lower bounds.
\end{remark}

\subsection*{Plan of the paper}

In Section~\ref{two} we show that any set-multilinear arithmetic
circuit of size $s$ can be efficiently transformed into an $O(\log s)$
depth set-multilinear circuit with unbounded fanin $+$ gates and fanin
$2$ $\times$ gates of size polynomial in $s$. This is quite similar to
the depth-reduction results of \cite{VSBR83,RY08} for more general
commutative circuits. As a result, size $s$ set-multilinear circuits
have $s^{O(\log s)}$ size set-multilinear branching programs.

In Section~\ref{three} we consider narrow set-multilinear branching
programs: a set-multilinear ABP computing a degree $d$ polynomial is
\emph{$w$-narrow} if in layer $d-w$ the number of distinct types of
the ABP is $O(w)$. We show that $n^{1/4}$-narrow ABPs computing the
Permanent requires $2^{\Omega(n^{1/4})}$ size. On the other hand,
a similar result for general set-multilinear ABPs
appears difficult. For instance, it would imply that the Permanent
requires superpolynomial size noncommutative arithmetic circuits,
which is an open problem for over two decades.

In Section~\ref{four}, we show that set-multilinear branching programs
are exponentially more powerful than \emph{interval} multilinear
circuits (where the index sets for each gate is restricted to be an
interval w.r.t.\ some ordering), assuming the sum-of-square conjecture
\cite{HWY10}. This further underlines the power of general
set-multilinear branching programs.

Finally, in Section~\ref{five} we consider set-multilinear circuits
with restrictions on the number of proof trees of monomials computed
by it, and prove exponential lower bounds results. This raises some
new lower bound questions.

\section{Depth Reduction of Set-Multilinear Circuits}\label{drcsm}\label{two}

We follow the standard method of depth reduction of commutative
arithmetic circuits \cite{VSBR83}, and use the exposition from Shpilka
and Yehudayoff's survey article \cite{SY10}. The general depth
reduction was adapted to syntactic multilinear circuits by Raz and
Yehudayoff \cite{RY08}. Our additional observation essentially is that
the depth reduction procedure can be carried out while preserving
set-multilinearity as well.

Given a commutative set-multilinear circuit $C$ of size $s$ computing
a set-multilinear polynomial $f$ of degree $d$ in the input variable
$X=X_1 \sqcup ...\sqcup X_d$, we show that there is another circuit
$C'$ of size $poly(s)$ and depth $O(\log d\log s)$ computing $f$.

\begin{theorem}\label{depth-red}
 Let $\Phi$ be a set-multilinear arithmetic circuit of size $s$ and
 degree $d$ over the field $\F$ and over the variable set $X$,
 partitioned as $X=X_1 \sqcup ...\sqcup X_d$, computing a polynomial
 $f\in\F[X]$. Then we can efficiently compute from $\Phi$ a
 set-multilinear arithmetic circuit $\Psi$, with multiplication gates
 of fanin $2$ and unbounded fanin $+$ gates, which is of size
 $O(s^3\log d)$ and depth $O(\log d)$ computing the polynomial $f$.
\end{theorem}

\begin{proof}
By definition, $\Phi$ is a homogeneous arithmetic circuit.  We assume
that $\Phi$ is non-redundant (i.e\ , for all gates $v$ in $\Phi$ the
polynomial $f_v$ computed at $v$ is nonzero). Since $\Phi$ is
set-multilinear, at each gate $v$ in $\Phi$ there is an associated
index set $I_v \subseteq [d]$ such that the polynomial $f_v$ is
set-multilinear of degree $|I_v|$ over the variable set $X_{I_v}$,
where
\[
X_{I_v}=\sqcup_{i\in I_v}X_i.
\]  

We denote the subcircuit rooted at the gate $v$ by $\Phi_v$.

\subsubsection*{Partial Derivative of $f_v$ by a gate $w$} 

Let $v, w$ be any two gates in circuit $\Phi$. Following the
exposition in \cite{SY10}, let $\Phi_{w=y}$ denote the circuit
obtained by removing any incoming edges at $w$ and labeling $w$ with a
new input variable $y$ and $f_{v,w}$ denote the polynomial (in $X\cup
\{y\}$) computed at gate $v$ in circuit $\Phi_{w=y}$. Define
\[
\partial_w f_v = \partial_y f_{v,w}.
\]
Note that $f_{v,w}$ is linear in $y$. 
Clearly, if $w$ does not occur in $\Phi_v$ then $\partial_wf_v=0$.  If $w$
occurs in $\Phi_v$, since $\Phi$ is set-multilinear the polynomial
$f_{v,w}$ is linear in $y$ and is of the form
\[
f_{v,w}=h_{v,w}y+g_{v,w}.
\]
Therefore, $\partial_wf_v=h_{v,w}$. We make the following immediate
observations from the set-multilinearity of $\Phi$.

\begin{itemize}
 \item Either $\partial_wf_v= 0$ or $\partial_wf_v$ is a homogeneous
   set-multilinear polynomial of degree $deg(v)-deg(w)$ over variable
   set $X \setminus X_{I_w}$.
\item If $deg(v) < 2.deg(w)$ and $v$ is a product gate with children
  $v_1,v_2$ such that $deg(v_1) \geq deg(v_2)$, then
  $\partial_wf_v=f_{v_2}.\partial_wf_{v_1}$.
\end{itemize}

For a positive integer $m$, let $G_m$ denote the set of product gates
$t$ with inputs $t_1,t_2$ in $\Phi$ such that $m < deg(t)$ and
$deg(t_1),deg(t_2) \leq m$. We observe the following claims (analogous
to \cite{SY10}) which are easily proved.

\begin{claim}\label{cl:c1}
 Let $\Phi$ be a set-multilinear nonredundant arithmetic circuit over
 variable set $X=\sqcup_{i=1}^d X_i$. Let $v$ be a gate in $\Phi$ such
 that $m < deg(v) \leq 2m$ for a positive integer $m$. Then
 $f_v=\sum_{t\in G_m}f_t.\partial_tf_v$.
\end{claim}

\begin{claim}\label{cl:c2}
 Let $\Phi$ be a set-multilinear non-redundant arithmetic circuit over
 the field $\F$ and over the set of variables $X$.  Let $v$ and $w$ be
 gates in $\Phi$ such that $0< deg(w) \leq m < deg(v) < 2deg(w)$. Then
 $\partial_wf_v=\sum_{t\in G_m}\partial_wf_t.\partial_tf_v$
\end{claim}

\textbf{Construction of $\Psi$:}

We now explain the construction of the depth-reduced circuit $\Psi$.
The construction is done in stages. Suppose upto Stage $i$ we have
computed, for $1\le j\le i$ the following:
\begin{itemize}
\item All polynomials $f_v$ for gates $v$ such that $2^{j-1}<deg(v)
  \leq 2^j$. 
\item All partial derivatives of the form $\partial_w f_v$ for gates
  $v$ and $w$ such that $2^{j-1}<deg(v)-deg(w) \leq 2^j$ and
  $deg(v)<2deg(w)$.
\item Furthermore, inductively assume that the circuit computed so far
  is set-multilinear of $O(i)$ depth, such that all product gates are
  fanin $2$, sum gates are of unbounded fanin.
\end{itemize}

We now describe Stage $i+1$ where we will compute all $f_v$ for gates
$v$ such that $2^i<deg(v) \leq 2^{i+1}$ and also all partial
derivatives of the form $\partial_w f_v$ for gates $v$ and $w$ such
that $2^i<deg(v)-deg(w) \leq 2^{i+1}$ and
$deg(v)<2deg(w)$. Furthermore, we will do this by adding a depth of
$O(1)$ to the circuit and $\poly(d,s)$ many new gates maintaining
set-multilinearity.\\

\noindent\textbf{Stage i+1:}~~We describe the construction at this
stage in two parts:\\

\textbf{Computing $f_v$}\\ 

Let $v$ be a gate in $\Phi$ such that $2^i < deg(v) \leq 2^{i+1}$ and
let $m=2^i$. By Claim \ref{cl:c1}, we have
\[
f_v=\sum_{t\in T}f_t\partial_tf_v=\sum_{t\in T}f_{t_1}f_{t_2}\partial_tf_v,
\]
where $T$ is the set of gates $t\in G_m$, with children $t_1$ and
$t_2$ such that $t$ is in $\Phi_v$. Note that if $t$ is not in
$\Phi_v$, then $\partial_tf_v=0$. Let $t\in T$ be a gate with inputs
$t_1$ and $t_2$. Thus. $m<deg(t)\leq 2m$, $deg(t_1) \leq m$, $deg(t_2)
\leq m$. Hence $deg(v)-deg(t)\leq 2^{i+1}-2^i=2^i$ and $deg(v)\leq
2^{i+1}<2.deg(t)$. Therefore, $f_{t_1},f_{t_2}$ and $\partial_tf_v$
are already computed. Thus, in order to compute $f_v$ we need $O(s)$
many $\times$ gates and $O(1)$ many $+$ gates. Overall, with $O(s^2)$
many new gates and $O(1)$ increase in depth we can compute all $f_v$
such that $2^i < deg(v) \leq 2^{i+1}$. Furthermore, we note that
$f_{t_1}$, $f_{t_2}$ and $\partial_tf_v$ are all set-multilinear
polynomials with disjoint index sets, and the union of their index
sets is $I_v$ for each $t\in T$. Thus, the new gates introduced all
preserve set-multilinearity.\\

\textbf{Computing $\partial_wf_v$}\\ 

Let $v$ and $w$ be gates in $\Phi$ such that $2^i < deg(v)-deg(w) \leq
2^{i+1}$ and $deg(v)<2deg(w)$. Let $m=2^i+deg(w)$. Thus, $deg(w)\leq m
<deg(v)<2deg(w)$. Note that $\partial_tf_v=0$ if $t\not\in T$. Hence
by Claim \ref{cl:c2} we can write 
\[
\partial_wf_v=\sum_{t\in T}\partial_wf_t\partial_tf_v,
\]
where $T$ is the set of gates in $\Phi_v$ that are contained in
$G_m$. For a gate $t\in T$, we have $deg(t)\le deg(v)<2deg(w)$.
Suppose $t_1$ and $t_2$ are the gates input to $t$ in the circuit
$\Phi$, and $\deg(t_1)\ge deg(t_2)$. Then we can write
\[
\partial_wf_v=\sum_{t\in T}f_{t_2}\partial_wf_{t_1}\partial_tf_v.
\]

We claim that $f_{t_2}$, $\partial_wf_{t_1}$, and $\partial_tf_v$ are
already computed.

\begin{itemize}
 \item Since $deg(v)\leq 2^{i+1}+deg(w)\le
   2^{i+1}+deg(t_1)=2^{i+1}+deg(t)-deg(t_2)$, we have $deg(t_2)\leq
   2^{i+1}+deg(t)-deg(v)\leq 2^{i+1}$. Hence $f_{t_2}$ is already
   computed (in first part of stage $i+1$).
\item Since $deg(t_1)-deg(w)\le 2^i$, the polynomial
  $\partial_wf_{t_1}$ is already computed in an earlier stage.
\item Since $deg(t)>m$, we have $deg(v)-deg(t)\le deg(v)-m\le
  2^{i+1}-2^i=2^i$. 
\item Thus, since $deg(v)\leq 2^{i+1}+deg(w)\leq 2(2^i+deg(w))<2deg(t)$, 
the polynomial $\partial_tf_v$ is already computed in an earlier stage.
\end{itemize}

As before, for each such pair of gates $w$ and $v$, we can compute
$\partial_wf_v$ with $O(s)$ new gates (using the polynomials already
computed in previous stages), and this increases the circuit depth by
$O(1)$. Doing this for all the pairs $(w,v)$ implies $O(s^3)$ new
gates are added in the process. Furthermore, the new gates included
clearly also have the set-multilinearity property.

At the end of the $\log d$ stages, the overall size of the resulting
set-multilinear circuit $\Psi$ is $O(s^3\log d)$ and its depth is
$O(\log d)$. This completes the proof of the theorem.
\end{proof}

\subsection*{Set-Multilinear Circuits to ABPs}

\begin{theorem}
 Given a set-multilinear arithmetic circuit of size $s$ and degree $d$
 over the field $\F$ and over the variable set $X=\sqcup_{i=1}^d X_i$,
 computing $f\in\F[X]$, we can transform it, in time $s^{O(\log d)}$,
 into a set-multilinear ABP of size $s^{O(\log d)}$ that computes $f$.
\end{theorem}

\begin{proof}
The proof of this theorem is fairly straightforward consequence of the
depth reduction result (Theorem~\ref{depth-red} in the previous
section). By Theorem~\ref{depth-red} we can assume to have computed a
circuit $\Psi$ of size $O(s^3\log d)$ and depth $O(\log d)$ for
computing $f$. By a standard procedure we transform the circuit $\Psi$
into a formula $F$ by duplicating the circuit at every gate. The
resulting formula is of size $s^{O(\log d)}$, because at every level
of the circuit there is a factor of $s$ increase in the size (as the
$+$ gates have unbounded fanin). The formula $F$ is clearly also
homogeneous, set-multilinear, and of depth $O(\log d)$. The formula
$F$ is also semi-unbounded: the product gates are fanin $2$ and plus
gates have unbounded fanin.

Next, we can apply a standard transformation (for e.g., see \cite{N91})
to transform the formula $F$ into a homogeneous algebraic branching
program (ABP). It is a bottom-up construction of the ABP: at a $+$
gate we can do a ``parallel composition'' of the input ABPs to
simulate the $+$ gate. At a $\times$ gate it is a sequential
composition of the two ABPs. Since the formula $F$ is set-multilinear,
the resulting ABP is also easily seen to be a set-multilinear ABP.
\end{proof}

\section{A Lower Bound Result for Set-Multilinear ABPs}\label{three}

As we have shown in Theorem~\ref{depth-red}, we can simulate
set-multilinear circuits of size $s$ and degree $d$ using
set-multilinear ABPs of size $s^{O(\log d)}$. Thus, proving even a
lower bound of $n^{\omega(\log n)}$ for set-multilinear ABPs computing
the $n\times n$ Permanent polynomial $\PER_n$ would imply
superpolynomial lower bounds for general set-multilinear circuits
computing $\PER_n$ which is a long-standing open problem.

However, in this section we show a lower bound result for
set-multilinear ABPs with restricted \emph{type width}, a notion that
we now formally introduce.

Let $P$ be a set-multilinear ABP computing a polynomial $f\in\F[X]$ of
degree $d$ with variable set $X=\sqcup_{i=1}^d X_i$. By definition,
the ABP $P$ is given by as layered directed acyclic graph with layers
numbered $0,1,\ldots,d$. Each node $v$ in layer $k$ of the ABP is
labeled by an index set $I_v\subseteq [d]$, and a degree $k$
set-multilinear polynomial $f_v$ over variables $\sqcup_{i\in I_v}X_i$
is computed at $v$ by the ABP. We refer to $I_v$ as the \emph{type} of
node $v$. The \emph{type width} of the ABP at layer $k$ is the number
of \emph{different} types labeling nodes at layer $k$.

The following proposition connects type width to the notion of
read-once oblivious ABPs (ROABPs defined in \cite{FS13}).

\begin{proposition}
Suppose $P$ is a set-multilinear ABP computing a polynomial
$f\in\F[X]$ of degree $d$ with variable set $X=\sqcup_{i=1}^d X_i$
such that the type-width of $P$ is $1$ at each layer. Then $P$ is in
fact an ROABP which is defined by a suitable permutation on the index
set $[d]$.
\end{proposition}

\begin{proof}
As each layer of $P$ has type width one, the list of type
$I_0=\emptyset\subset I_1\subset\dots\subset I_d$ gives an ordering of
the index set, where the $i^{th}$ index in the ordering is
$I_i\setminus I_{i-1}$. W.r.t.\ this ordering clearly $P$ is an ROABP.
\end{proof}

It is well-known that Nisan's rank argument \cite{N91} (originally
used for lower bounding noncommutative ABP size) also yields
exponential lower bounds for any ROABP computing $\PER_n$. In
particular, it implies an exponential lower bound for set-multilinear
ABPs of type-width $1$.

We can formulate a general rank-based approach for set-multilinear
ABPs.

Let $P$ be a set-multilinear ABP computing a polynomial $f\in\F[X]$ of
degree $d$ with variable set $X=\sqcup_{i=1}^d X_i$, with layers of
the ABP numbered $0,1,\ldots,d$. For simplicity assume $|X_i|=n$ for
each $i$. For each index set $I\subseteq [d]$ let $M_I$ denote the set
of all monomials of the form $\prod_{i\in I}x_{ij}$, where $x_{ij}\in
 X_i$ for each $i\in I$.

For every monomial $m\in M_{[d]}$ let 
\[
S_m=\{(m_1,m_2)\mid m=m_1m_2, m_1\in M_I, m_2\in M_{[d]\setminus I}\textrm{ for
some }I\in{[d]\choose k}\}.
\]
For each $k$, we consider matrices $M_k$ whose rows are labeled by
monomials from
\[
\sqcup_{I\in{[d]\choose k}}M_I,
\]
and columns are labeled by monomials from 
\[
\sqcup_{I\in{[d]\choose k}}M_{[d]\setminus I}.
\]
with the property that for every monomial $m\in M_{[d]}$ its coefficient
$f(m)$ in $f$ satisfies

\begin{equation}\label{mk-property}
f(m) = \sum_{(m_1,m_2)\in S_m} M_k(m_1,m_2).
\end{equation} 

For each $k$, let $rank_k(f)$ denote be the minimum rank attained by
the $rank(M_k)$ for any matrix $M_k$ satisfying the above property.
We have the following lower bound on the size of set-multilinear ABPs
computing $f$, following Nisan's rank argument~\cite{N91}.

\begin{theorem}\label{rank-smABPsize}
Let $f\in\F[X]$ be a set-multilinear polynomial of degree $d$ with
variable set $X=\sqcup_{i=1}^d X_i$. Then any set-multilinear ABP
computing $f$ is of size at least $\sum_{k=0}^d rank_k(f)$.
\end{theorem}

\begin{proof}
Let $P$ be a minimum size set-multilinear ABP computing the polynomial
$f$. Define matrices $L_k$ and $R_k$ as follows.  Let
$v_1,v_2,...,v_r$ be the nodes in the $k^{th}$ layer of the ABP. Label
the rows of matrix $L_k$ by degree-$k$ monomials from $M_I$ for
$I\in{[d]\choose k}$, and the columns of $L_k$ by $v_1,v_2,...,v_r$.
The entry $L_k[m_1,v_i]$ is defined as the coefficient of monomial
$m_1$ in the polynomial computed at node $v_i$. The rows of $R_k$ are
labelled $v_1,v_2,...,v_r$ and columns by degree $d-k$ monomials from
$M_{[d]\setminus I}$ for $I\in{[d]\choose k}$. The entry
$R_k[v_i,m_2]$ is defined as the coefficient of the monomial $m_2$ in
the (set-multilinear) ABP computed between node $v_i$ and the sink
node of the ABP.

By construction the product matrix $L_kR_k$ clearly satisfies
Equation~\ref{mk-property}. The claim now follows, since for each $0
\leq k \leq d$, we have
\[
r \geq rank(L_k) \geq min\{rank(L_k),rank(R_k)\}
\geq rank_k(f). 
\]
Therefore $\sum_{k=0}^d rank_k(f)$ lower bounds the size of $P$.
\end{proof}

\subsection{Lower bounds for narrow set-multilinear ABPs}

\begin{definition}
A set-multilinear ABP computing a degree $d$ polynomial in $\F[X]$
such that $X=\sqcup_{i=1}^d X_i$ is said to be \emph{$w(d)$-narrow} if
the type-width of the ABP at layer $d-w(d)$ is $O(w(d))$.
\end{definition}

\begin{theorem}\label{lb-smabp}
Any $(n^{1/4}+o(n^{1/4}))$-narrow set-multilinear ABP computing the
permanent polynomial $\PER_n$ requires size
$2^{\Omega(n^{1/4})}$.\footnote{The same lower bound proof will work
  for the Determinant polynomial $\DET_n$.}
\end{theorem}

\begin{proof}
Let $P$ be a set-multilinear $n^{1/4}+o(n^{1/4})$-narrow ABP computing
$\PER_n$.  Let $k = n^{1/4}+o(n^{1/4})$ be the layer with type-width
$O(n^{1/4})$, and $V_k$ denote the set of nodes in the $k^{th}$ layer
of $P$. For each node $v\in V_k$ let $I_v$ denote its index set. We
know that $|I_v|=k$.

\begin{claim}
 $|\bigcap_{v \in V_k}I_v|=n-O(\sqrt{n})$.
\end{claim}

\begin{proof}
We note that
\begin{eqnarray*}
|\bigcap_{v \in V_k}I_v| &=& n-|\bigcup_{v\in V_k}\overline{I_v}|\\ 
&\ge & n - \sum_{v\in V_k}|\overline{I_v}|\\ 
& \ge & n - cn^{1/4}|V_k|\\
& = & n- O(n^{1/2}).
\end{eqnarray*}
\end{proof}

Let $S=\bigcup_{v\in V_k}\overline{I_v}$. Note that $|S|=O(\sqrt{n})$,
as already used in the above claim. Likewise, $X=\bigcap_{v \in
  V_k}I_v$ is of size $|X|=n-O(\sqrt{n})$. Fix a constant $\alpha>0$
such that $\alpha\sqrt{n}>|S|$. Now, we choose an index set $Z$
containing $S$ such that $|Z|=\alpha\sqrt{n}$, and choose $Y \subseteq
X\setminus Z$ such that $|Y|=\alpha\sqrt{n}$.

We are going to focus on the permanent of the $2\alpha\sqrt{n}\times
2\alpha\sqrt{n}$ submatrix consisting of variable $\{X_{ij}\mid i,j\in
  Y\cup Z\}$. To that end, we first set some variables of $\PER_n$ to
  constants as follows:

\[
    X_{ij}= 
\begin{cases}
    1,              & \text{if } i\notin Y\cup Z \text{ and } i= j\\
    0,              & \text{if } i\notin Y\cup Z\text{ and } i\neq j\\
\end{cases}
\]


The resulting ABP, after these substitutions, clearly computes
$\PER_{2\alpha\sqrt{n}}$ on the variables corresponding to index set $Y
\cup Z$. Furthermore, each node $v\in V_k$ now computes a homogeneous
set-multilinear polynomial of the same degree $2\alpha\sqrt{n}- r$,
where $r=n^{1/4}+o(n^{1/4})$ is the same for each $v\in V_k$.

Now, we further simplify entries in the submatrix indexed by $Y\cup
Z$. In the sequel let $\nu=\alpha\sqrt{n}$, and let 

\begin{eqnarray*}
Y & = &\{i_1,i_2\ldots,i_\nu\}, \text{ and}\\
Z & = &\{j_1,j_2\ldots,j_\nu\}.
\end{eqnarray*}

In the submatrix indexed by $Y\cup Z$ we renumber rows and columns so
that the indexing order is $i_1,j_1,i_2,j_2,\ldots,i_\nu,j_\nu$. We
set to zero all variables $X_{ij}$ if $i\ne j$ and
$\{i,j\}\ne\{i_s,j_s\}$ for some $1\le s\le \nu$. The resulting
  submatrix is now block diagonal with $2\times 2$ blocks along the
  diagonal, each block indexed by variables $\{i_s,j_s\}$ for $1\le
    s\le \nu$.  In all there are $2^\nu$ many nonzero monomials in the
    permanent of this submatrix.  After these simplifications, the
    resulting ABP computes the permanent of this block diagonal
    matrix. Let $A$ denote this block diagonal matrix.

Let $\mathcal{M}$ denote the set of \emph{all}, $2^\nu$ many, degree
$\nu$ monomials of the form $\prod_{a\in Y,b\in S} X_{ab}$, where
$S\in {Y\cup Z\choose r}$.  Likewise, let $\mathcal{M}'$ denote the
set of $2^\nu$ many degree $\nu$ monomials of the form $\prod_{a\in
  Z,b\in T} X_{ab}$, where $T\in {Y\cup Z\choose r}$.

We now define two matrices $L_k$ and $R_k$ as follows: The rows of
$L_k$ are labeled by monomials in $\mathcal{M}$. The columns of $L_k$
are labeled by pairs $(v,m')$, where $v\in V_k$ and $m'$ is a monomial
of degree $\nu-r$ on variables $X_{ij}$, where $i\in Y\cup Z$.

The entries of $L_k$ are defined as follows: $L_k(m,(v,m'))$ is the
coefficient of monomial $mm'$ in polynomial computed by ABP at node
$v$.

The matrix $R_k$ has its rows labeled by pairs $(v,m')$, where $v\in
V_k$ and $m'$ is a monomial of degree $\nu-r$ on variables $X_{ij}$,
where $i\in Y\cup Z$. The columns of $R_k$ are labeled by the
monomials in $\mathcal{M}'$.

The entries of $R_k$ are defined as follows: If $m'$ does not divide
$m$ then $R_k((v,m'),m)=0$. If $m'$ divides $m$ then $R_k((v,m'),m)$
is defined as the coefficient of the monomial $m$ in the ABP computed
between node $v$ and the sink node of the ABP.

Thus, the product matrix $L_kR_k$ is a $2^\nu\times 2^\nu$ matrix
whose rows are labeled by monomials from $\mathcal{M}$ and columns by
monomials from $\mathcal{M}'$. 

By assumption, the simplified ABP computes the permanent of block
diagonal matrix $A$ described above. Hence, in the product matrix
$L_kR_k$ the $(m_1,m_2)^{th}$ entry is the coefficient of the monomial
$m_1m_2$ in $Perm(A)$. 

We now lower bound the rank of $L_kR_k$ following Nisan's argument
\cite{N91}. For subsets $S,T\subset Y\cup Z$ such that
$S\cap\{i_s,j_s\}=1$ and $T\cap\{i_s,j_s\}=1$, $1\le s\le \nu$,
consider the $(S,T)^{th}$ submatrix of $L_kR_k$ whose rows are labeled
by degree $\nu$ monomials of the form $\prod_{i\in Y,j\in S}X_{ij}$
and columns by degree $\nu$ monomials of the form $\prod_{i\in Z,j\in
  T}X_{ij}$. We observe that the $(S,T)$ submatrix has nonzero entries
precisely when $S\cap T=\emptyset$. Hence the rank of $L_kR_k$ is
lower bounded by $2^\nu$, which is the number of choices of $S$. On
the other hand, the rank of $L_kR_k$ is upper bounded by the number of
columns of $L_K$ which is $|V_k|2^{\nu-r}$. Hence we have
\[
2^\nu\le rank(L_kR_k) \le |V_k|2^{\nu-r}.
\]

Since $r=n^{1/4}-o(n^{1/4})$, it follows that $rank(L_kR_k)\ge
2^r=2^{\Omega(n^{1/4})}$ which completes the proof.

\end{proof}

\begin{remark}
 The proof of the above theorem can be easily modified to show a more
 general result: Any $n^{1/2-\epsilon}$-narrow set-multilinear ABP for
 $\PER_n$ requires size $2^{n^{\Omega(1)}}$.
\end{remark}

As a consequence of Theorem~\ref{lb-smabp} we immediately obtain the
following lower bound on the size of a sum of $n^{1/2-\epsilon}$ many
ROABPs for computing the permanent.

\begin{corollary}
\label{cor:lb-smabp}
Let $P_i, 1\le i\le r$ be ROABPs such
that $\sum_{i=1}^rP_i$ is the permanent polynomial $\PER_n$ (or for
the determinant polynomial $\DET_n$). If $r\le n^{1/2-\epsilon}$ then
at least one of the $P_i$ is of size $2^{n^{\Omega(1)}}$.
\end{corollary}

\section{Interval multilinear circuits and ABPs}\label{four}

For variable partition $X=\sqcup_{i=1}^dX_i$ let $f\in\F[X]$ be a
set-multilinear polynomial.

For a permutation $\sigma\in S_d$, a \emph{$\sigma$-interval
  multilinear circuit} $C$ for computing $f$ is a special kind of
set-multilinear arithmetic circuit: for every gate of the circuit the
corresponding index set is a \emph{$\sigma$-interval}
$\{\sigma(i),\sigma(i+1),\ldots,\sigma(j)\}$, $1\le i\le j\le d$.

Similarly, a \emph{$\sigma$-interval multilinear ABP} is a
set-multilinear ABP such that the index set associated to every node
is some $\sigma$-interval.

The aim of the present section is to compare the computational power
of interval multilinear circuits with general set-multilinear
circuits. It is clear that $\sigma$-interval multilinear circuits are
restricted by the ordering of the indices. In essence,
$\sigma$-interval multilinear circuits are restricted to compute like
noncommutative circuits (with respect to the ordering prescribed by
$\sigma$). This property needs to be exploited to prove the
separations. We show the following two results.

\begin{enumerate}
\item We show that there are monotone set-multilinear polynomials
  $f\in\mathbb{R}[X]$ that have monotone set-multilinear circuits
  (even ABPs) of linear in $d$ size, but requires $2^{\Omega(d)}$ size
  monotone $\sigma$-interval multilinear circuits for every $\sigma\in
  S_d$.

\item Assuming the sum-of-squares conjecture \cite{HWY10} we show that
  the polynomials constructed for the above even require
  $2^{\Omega(d)}$ size general $\sigma$-interval multilinear circuits
  for every $\sigma\in S_d$.
\end{enumerate}

A new aspect is that the polynomial we construct is only partially
explicit. We use the probabilistic method to pick certain parameters
that define the polynomial.

\subsection*{The polynomial construction}

Let $X=\sqcup_{i=1}^{2d}X_i$ be the variable set, where
\[
X_i = \{x_{0,i},x_{1,i}\}, 1\le i\le 2d.
\]

For every binary string $b\in\{0,1\}^d$ we define the monomials:

\begin{eqnarray*}
w_b &=& \prod_{i=1}^dx_{b_i,i}\\
w'_b&=&\prod_{i=1}^{d}x_{b_i,d+i}.
\end{eqnarray*}

We define the set-multilinear polynomial $P \in\F[X]$ as follows:
\[
P = \sum_{\overline{b}\in\{0,1\}^d}w_bw'_b.
\]

For any permutation $\sigma\in S_{2d}$ permuting the indices in $[2d]$,
we define monomials:
\begin{eqnarray*}
\sigma(w_b) & = & \prod_{i=1}^dx_{b_i,\sigma(i)}\\
\sigma(w'_b) & = &\prod_{i=1}^{d}x_{b_i,\sigma(d+i)},
\end{eqnarray*}

and the corresponding polynomial $\sigma(P)$ as follows:
\[
\sigma(P) = \sum_{\overline{b}\in\{0,1\}^d}\sigma(w_b)\sigma(w'_b).
\]

\begin{remark}
In the polynomial $\sigma(P)$ we refer to the indices $\sigma(j)$ and
$\sigma(d+j)$ as a \emph{matched pair of indices}, since in the
monomial $\sigma(w_b)\sigma(w'_b)$ it is required that the variables
$x_{b_j,\sigma(j)}$ and $x_{b_j,\sigma(d+j)}$ have the same first
index $b_j$. For $\sigma=id$, the matched pairs are $(j,d+j)$ for
$1\le j\le d$.
\end{remark}

\begin{lemma}\label{ppi-abp}
The set-multilinear polynomial $\sigma(P)$ can be computed by a
monotone set-multilinear ABP of size $O(d)$.
\end{lemma}

\begin{proof}
The set-multilinear ABP computes the polynomial at layer $2$
\[
x_{0,\sigma(1)}x_{0,\sigma(d+1)}+x_{1,\sigma(1)}x_{1,\sigma(d+1)},
\]
where the index set at that layer is $\{\sigma(1),\sigma(d+1)\}$. At
the $2i^{th}$ layer suppose the polynomial computed by the ABP is
\[
\sigma(P)^{(i)} =
\sum_{\overline{u}\in\{0,1\}^i}\prod_{j=1}^ix_{u_j,\sigma(j)}\prod_{j=1}^ix_{u_j,\sigma(d+j)},
\]
where the index set is
$\{\sigma(1),\sigma(2),\ldots,\sigma(i)\}\sqcup\{\sigma(d+1),\sigma(d+2),
\ldots,\sigma(d+i)\}$. Then at layer $2i+2$ we can compute
\[
\sigma(P)^{(i+1)} = P^{(i)} x_{0,\sigma(i+1)}x_{0,\sigma(d+i+1)} + P^{(i)}
x_{1,\sigma(i+1)}x_{1,\sigma(d+i+1)}.
\]
We observe that the ABP remains set-multilinear. Furthermore, it is
monotone. Clearly, at layer $2d$ the ABP computes the polynomial
$\sigma(P)$ as desired.
\end{proof}

An immediate consequence is the following corollary.

\begin{corollary}
For any collection of permutations
$\sigma_1,\sigma_2,\ldots,\sigma_s\in S_{2d}$ the polynomial
$\sum_{j=1}^s \sigma_i(P)$ can be computed by a monotone
set-multilinear ABP of size $O(sd)$.
\end{corollary}

We will show that there exist a set of $d$ permutations
$\sigma_1,\sigma_2,\ldots,\sigma_d\in S_{2d}$ with the following
property: for any permutation $\tau\in S_{2d}$ there is a $\sigma_i$
from this set such that any $\tau$-interval multilinear circuit that
computes $\sigma_i(P)$ requires size $2^{\Omega(d)}$.

Note that, in contrast, given a $\sigma(P)$ there is always a
permutation $\tau\in S_{2d}$ such that $\sigma(P)$ can be computed by
a small $\tau$-interval multilinear circuit. Indeed, the ABP computing
$\sigma(P)$ described in Lemma~\ref{ppi-abp} is an
interval-multilinear ABP w.r.t.\ the ordering
$\sigma(1),\sigma(d+1),\sigma(2),\sigma(d+2),\ldots,\sigma(d),\sigma(2d)$.

\subsection*{Interval multilinear circuits and noncommutative circuits}

We first observe that a $\tau$-interval multilinear circuit computing
a set-multilinear polynomial in $\F[X]$, $X=\sqcup_{i=1}^d X_i$, is
essentially like a noncommutative circuit computing a noncommutative
polynomial over the variables $X$, whose monomials can be considered
as words of the form $x_{i_1}x_{i_2}\ldots x_{i_d}$, where $x_{i_j}\in
X_{\tau(j)}$ for $1\le j\le d$.

In \cite{HWY10}, Hrubes et al have related the well-known
sum-of-squares (in short, SOS) conjecture (also see \cite{Shap}) to
lower bounds for \emph{noncommutative arithmetic circuits}. Our
results in this section are based on their work. We recall the
conjecture.\\

\noindent\textbf{The sum-of-squares (SOS) conjecture:} Consider the
question of expressing the biquadratic polynomial
\[
SOS_k(x_1,\ldots,x_k,y_1,\ldots,x_k)=(\sum_{i\in
  [k]}x_i^2)(\sum_{i\in [k]}y_i^2)
\]
 as a sum of squares $(\sum_{i\in [s]}f_i^2)$ for the least possible
 $s$, where each $f_i$ is a homogeneous bilinear polynomial. The
 conjecture states that $s=\Omega(k^{1+\epsilon})$ over the field of
 complex numbers $\mathbb{C}$ (or the algebraic closure of any field
 $\F$ such that $char(\F)\ne 2$).

The following lower bound is shown in \cite{HWY10} assuming the SOS
conjecture.

\begin{theorem}{\rm\cite{HWY10}}\label{hwy-thm}
Assuming the SOS conjecture over field $\F$, any noncommutative
circuit computing the polynomial $ID=\sum_{w\in\{x_0,x_1\}^d}ww$ in
\emph{noncommuting variables} $x_0$ and $x_1$ requires size
$2^{\Omega(d)}$.
\end{theorem}

It immediately implies the following conditional lower bound for
interval multilinear circuits.

\begin{corollary}\label{cor1}
Assuming the SOS conjecture, for any $\sigma\in S_{2d}$ a
$\sigma$-interval multilinear circuit computing the set-multilinear
polynomial $\sigma(P)$ requires size $2^{\Omega(d)}$.
\end{corollary}

\begin{remark}
The connection between the SOS conjecture and lower bounds shown in
\cite{HWY10} for the noncommutative polynomial
$\sum_{w\in\{x_0,x_1\}^d}ww$ is made by considering the polynomial as
$\sum_{w_1,w_2\in\{x_0,x_1\}^{d/2}}w_1w_2w_1w_2$. Then the words $w_1$
and $w_2$ are treated as single variables, and allowing $w_1$ and
$w_2$ to commute transforms the polynomial into $(\sum w_1^2)(\sum
w_2^2)$. If $\sum_{w\in\{x_0,x_1\}^d}ww$ has a $2^{o(d)}$ size
noncommutative circuit, then it turns out that $(\sum w_1^2)(\sum
w_2^2)$ can be written as a small sum-of-squares, contradicting the
conjecture.
\end{remark}

{From} the above remark we can deduce the following corollary which is
stronger than Corollary~\ref{cor1}.

\begin{corollary}\label{cor2}
For even $d$, partition $[2d]$ into four intervals of size $d/2$ each:
$I_1=[1\ldots d/2]$, $I_2=[d/2+1\ldots d]$, $I_3=[d+1\ldots 3d/2]$,
and $I_4=[3d/2+1\ldots 2d]$. Let $\sigma\in S_{2d}$ be any permutation
such that $\sigma(I_j)=I_j, 1\le j\le 4$. Then, assuming the SOS
conjecture, any $id$-interval multilinear circuit computing the
polynomial $\sigma(P)$ requires size $2^{\Omega(d)}$.
\end{corollary}

\begin{proof}
By definition, $\sigma(P)=\sum_{b\in\{0,1\}^d}\sigma(w_b)\sigma(w'_b)$. Since
$\sigma$ stabilizes each $I_j, 1\le j\le 4$, we can write
$\sigma(w_b)=w_1w_2$ and $\sigma(w'_b)=w'_1w'_2$, where the matched
pairs are between $w_1$ and $w'_1$, and between $w_2$ and $w'_2$,
respectively.  Now, by substituting the same variable for each matched
pair, we can obtain the polynomial $(\sum w_1^2)(\sum w_2^2)$. As
explained in the proof of Corollary~\ref{cor1}, the argument in
\cite{HWY10} can now be applied to yield the circuit size lower bound
of $2^{\Omega(d)}$, assuming the SOS conjecture for the biquadratic
polynomial $(\sum w_1^2)(\sum w_2^2)$, treating each $w_1$ and $w_2$
as individual variables.
\end{proof}


We will use the probabilistic method to show the existence of the set
of permutations $\sigma_i, 1\le i\le d$ in $S_{2d}$, such that for
each $\tau\in S_{2d}$ there is some $\sigma_i(P), i\in [d]$ that
requires size $2^{\Omega(d)}$ $\tau$-interval multilinear
circuits. We will require the following concentration bound.

\begin{theorem}{\rm\cite[Theorem 5.3, page 68]{dp09}}\label{thm:mabd}
Let $X_1,\cdots,X_n$ be any $n$ random variables and let $f$ be a
function of $X_1,X_2\ldots,X_n$. Suppose for each $i \in [n]$ there is
$c_i\ge 0$ such that
\[
|\mathbb{E}[f|X_1,\cdots,X_i]-\mathbb{E}[f|X_1,\cdots,X_{i-1}]|\le c_i. 
\]
Then for any $t>0$, we have the bound $\mathrm{Prob}[f <
  \mathbb{E}[f]-t]\le exp(-\frac{2t^2}{c})$, where $c=\sum_{i \in
  [n]}c_i^2$.
\end{theorem}

\begin{lemma}\label{whp-thm}
Let $\sigma\in S_{2d}$ be a permutation picked uniformly at random.
For any $\tau\in S_{2d}$, the probability that $\sigma(P)$ is
computable by a $\tau$-interval multilinear circuit of size $2^{o(d)}$
is bounded by $e^{-\Omega(d)}$, assuming the SOS conjecture.
\end{lemma}


\begin{proof}
In the polynomial $P=\sum_{b\in\{0,1\}^d}w_bw'_b$ the \emph{matched
  pairs}, as defined earlier, are $(i,d+i), 1\le i\le d$. As in
Corollary~\ref{cor2}, we partition the index set $[2d]$ into four
consecutive $d/2$-size intervals $I_1=[1\ldots d/2]$,
$I_2=[d/2+1\ldots d]$, $I_3=[d+1\ldots 3d/2]$, and $I_4=[3d/2+1\ldots
  2d]$. Note that $d/2$ of the matched pairs are between $I_1$ and
$I_3$ and the remaining $d/2$ between $I_2$ and $I_4$. Consider
the following two subsets of matched pairs of size $d/8$ each:

\begin{eqnarray*}
E_1 & = & \{(i,d+i)\mid 1\le i\le d/8\}\\
E_2 & = & \{(d/2+i,3d/2+i)\mid 1\le i\le d/8\}.
\end{eqnarray*}

The pairs in $E_1$ are between $I_1$ and $I_3$ and pairs in $E_2$ are
between $I_2$ and $I_4$. Let $\sigma\in S_{2d}$ be a permutation
picked uniformly at random. We say $(i,d+i)\in E_1$ is \emph{good} if
$\sigma(i)\in I_1$ and $\sigma(d+i)\in I_3$. Similarly,
$(d/2+i,3d/2+i)\in E_2$ is called \emph{good} if $\sigma(d/2+i)\in
I_2$ and $\sigma(3d/2+i)\in I_4$. Let $X_i$, $1\le i\le d/8$, be
indicator random variables which take the value $1$ iff the edge
$(i,d+i)\in E_1$ is good. Similarly, define indicator random variables
$X'_i$ corresponding to $((d/2+i,3d/2+i)\in E_2, 1\le i\le d/8$. Note
that

\begin{eqnarray*}
1/64 & \le & \mathrm{Prob}_{\sigma\in S_{2d}}[X_i=1]  \le  1/16\\
1/64 & \le & \mathrm{Prob}_{\sigma\in S_{2d}}[X'_i=1] \le  1/16.
\end{eqnarray*}

Let $f=\sum_{i=1}^{d/8}X_i$ and $f'=\sum_{i=1}^{d/8}X'_i$. Clearly,

\begin{eqnarray*}
d/512 & \le & \mathrm{E}[f]  \le  d/128\\
d/512 & \le & \mathrm{E}[f'] \le  d/128.
\end{eqnarray*}

Furthermore, we also have for each $i: 1\le i\le d/8$

\begin{eqnarray*}
|\mathbb{E}[f|X_1,X_2,\ldots,X_i] - \mathbb{E}[f|X_1,X_2,\ldots,X_{i-1}]|\le 1/16\\
|\mathbb{E}[f'|X'_1,X'_2,\ldots,X'_i] - \mathbb{E}[f'|X'_1,X'_2,\ldots,X'_{i-1}]|\le 1/16.
\end{eqnarray*}

Applying Theorem~\ref{thm:mabd} we deduce that

\begin{eqnarray*}
\mathrm{Prob}_{\sigma\in S_{2d}}[f<{\frac{d}{1024}}] & \le & e^{-\alpha d}\\
\mathrm{Prob}_{\sigma\in S_{2d}}[f'<{\frac{d}{1024}}] & \le & e^{-\alpha d},
\end{eqnarray*}

where $\alpha>0$ is some constant independent of $d$.  Hence,

\[
\mathrm{Prob}_{\sigma\in S_{2d}}[f\ge {\frac{d}{1024}}\textrm{ and } f'\ge
  {\frac{d}{1024}}] \ge 1 -2 e^{-\alpha d}.
\]

Thus, with probability $1-2e^{-\alpha d}$ there are $d/1024$ pairs
$(\sigma(i),\sigma(d+i))$ such that $\sigma(i)\in I_1$ and
$\sigma(d+i)\in I_3$, and there are $d/1024$ pairs
$(\sigma(d/2+i),\sigma(3d/2+i))$ such that $\sigma(d/2+i)\in I_2$ and
$\sigma(3d/2+i)\in I_4$. If we set all other variables in the
polynomial $\sigma(P)$ to $1$, we can apply Corollary~\ref{cor2} to
the resulting polynomial (with $d$ replaced by $d/1024$ in the
corollary) which will yield the lower bound of $2^{\Omega(d)}$ for any
$id$-interval multilinear circuit (here $id$-interval multilinear
circuit means interval-multilinear with respect to the standard
ordering $\{1,2\ldots,d\}$) computing $\sigma(P)$ with probability $1
  -2 e^{-\alpha d}$. For any $\tau$-interval multilinear circuit too
  the same lower bound applies because $\tau\sigma$ is also a random
  polynomial in $S_{2d}$ with uniform distribution.
\end{proof}

We are now ready to state and prove the main result of this
section.

\begin{theorem}\label{sm-im-sep}
There is a set-multilinear polynomial $f\in\F[X]$, where
$X=\sqcup_{i=1}^{\log d + 2d}X_i$ and $X_i = \{x_{0,i},x_{1,i}\}, 1\le
i\le \log d + 2d$ such that $f$ has an $O(d^2)$ size monotone
set-multilinear ABP and, assuming the SOS conjecture, for any $\tau\in
S_{\log d + 2d}$ any $\tau$-interval multilinear circuit computing $f$
has size $2^{\Omega(d)}$.
\end{theorem}

\begin{proof}
By Lemma~\ref{whp-thm}, if we pick permutations
$\sigma_1,\sigma_2,\ldots,\sigma_d\in S_{2d}$ independently and
uniformly at random then, assuming the SOS conjecture, the probability
that every $\sigma_i(P)$ can be computed by some $\tau$-interval
multilinear circuit is bounded by $2^de^{-\alpha
  d^2}=e^{-\Omega(d^2)}$. As there are only $(2d)!$ many permutations
$\tau$, by the union bound it follows that there exist permutations
$\sigma_1,\sigma_2,\ldots,\sigma_d\in S_{2d}$ such that for any $\tau$
at least one of the $\sigma_i(P)$ requires $2^{\Omega(d)}$ size
$\tau$-interval multilinear circuits. We will define the polynomial
$f$ by ``interpolating'' the $\sigma_i(P)$, and we need the fresh
$\log d$ variable sets in order to do the interpolation. For each $c:
1\le c\le d$, let its binary encoding also be denoted by $c$, where
$c\in\{0,1\}^{\log d}$. Let $u_c$ denote the monomial
\[
u_i=\prod_{j=2d+1}^{2d+\log d}x_{c_j,j}.
\]

Hence the monomial $u_c$ can also be seen as an encoding of $c: 1\le
c\le d$.

We define the polynomial $f$ as
\[
f = \sum_{c\in\{0,1\}^{\log d}}u_c\sigma_c(P).
\]
Clearly, $f\in\F[X]$ and for each $0$-$1$ assignment to the variables
in $X_j, 2d+1\le j\le 2d+\log d$, the polynomial $f$ becomes
$\sigma_c(P)$.

Now, if $f$ had a $2^{o(d)}$ size $\tau$-interval multilinear circuit
for any $\tau\in S_{2d}$, by different $0$-$1$ assignments to
variables in $X_j, 2d+1\le j\le 2d+\log d$ we will obtain $2^{o(d)}$
size $\tau$-interval multilinear circuit for each $\sigma_c(P)$ which
is a contradiction.
\end{proof}

Finally, we note that for monotone $\tau$-interval multilinear
circuits, the lower bound results of Theorem~\ref{hwy-thm},
Corollaries \ref{cor1} and \ref{cor2}, and Lemma~\ref{whp-thm} hold
unconditionally (without assuming the SOS conjecture) by a direct rank
argument. As a consequence we have the following.

\begin{theorem}\label{monotone-sep}
There is a set-multilinear polynomial $f\in\F[X]$, where
$X=\sqcup_{i=1}^{\log d + 2d}X_i$ and $X_i = \{x_{0,i},x_{1,i}\}, 1\le
i\le \log d + 2d$ such that for any $\tau\in S_{\log d + 2d}$ any
monotone $\tau$-interval multilinear circuit computing $f$ has size
$2^{\Omega(d)}$.
\end{theorem}

\begin{proof}
The polynomial $f$ defined in Theorem~\ref{sm-im-sep}, which has small
monotone set-multilinear ABPs, requires $2^{\Omega(d)}$ size monotone
$\tau$-interval multilinear circuits for all $\tau\in S_{2d+\log d}$,
as explained above.
\end{proof}

\section{Proof tree restrictions on set-multilinear circuits}\label{five}

In this section we study set-multilinear circuits that satisfy a
``semantic'' restriction on the proof trees of monomials computed by
them. Specifically, we consider two such restrictions and show
superpolynomial lower bounds for set-multilinear circuits with such
restrictions computing the Permanent.

For an arithmetic circuit $C$, a \emph{proof tree} for a monomial $m$
is a multiplicative subcircuit of $C$ rooted at the output gate
defined by the following process starting from the output gate:

\begin{itemize}
\item At each $+$ gate retain exactly one of its input gates.
\item At each $\times$ gate retain both its input gates.
\item Retain all inputs that are reached by this process.
\item The resulting subcircuit is multiplicative and computes the
  monomial $m$ (with some coefficient).
\end{itemize}

Let $C$ be a set-multilinear circuit computing $f\in\F[X]$ for
variable partition $X=\sqcup_{i=1}^d X_i$. Then any proof tree for a
monomial is, in fact, a \emph{binary tree} with leaves labeled by
variables (ignoring the leaves labeled by constants) and internal
nodes labeled by gate names (the $\times$ gates of $C$ occurring in
the proof tree). By set-multilinearity, in each proof tree there is
exactly one variable from each subset $X_i$, and each variable occurs
at most once in a proof tree. 

\begin{definition}
A \emph{proof tree type} $T$ for a set-multilinear circuit $C$
computing a degree $d$ polynomial in $\F[X]$ is a binary tree with $d$
leaves. Each node $v$ of $T$ is labeled by an index set $I_v\subseteq
[d]$: The root is labeled by $[d]$, each leaf is labeled by a distinct
singleton set $[i], 1\le i\le d$, and if $v$ has children $v_1$ and
$v_2$ in the tree then $I_v=I_{v_1}\sqcup I_{v_2}$. For any tree type
$T$, the corresponding \emph{truncated tree type} $\hat{T}$ is the
subtree of $T$ obtained by deleting all nodes $v$ such that $|I_v|\le
d/3$. Notice that any truncated tree type $\hat{T}$ is a binary tree
with \emph{at most} two leaves (although its depth could be
unbounded).
\end{definition}

We introduce some notation. Let $C$ be a set-multilinear circuit for
variable partition $X=\sqcup_{i=1}^d X_i$. Let $m\in X^d$ be any
degree $d$ monomial.

\begin{itemize}
\item The set of all proof trees of $m$ in circuit $C$ is denoted
  $\mathcal{P}_{C,m}$.
\item The set of all proof tree types of $m$ in circuit $C$ is denoted
  $\mathcal{T}_{C,m}$.
\item The set of all truncated proof tree types of $m$ in circuit $C$
  is denoted $\hat{\mathcal{T}}_{C,m}$.
\end{itemize}

Note that, in general, monomial $m$ can have many proof trees. For
every proof tree in $\mathcal{P}_{C,m}$ there is a corresponding proof
tree type in $\mathcal{T}_{C,m}$ obtained by dropping variable and
gate names from the proof tree and labeling the nodes instead with the
corresponding index sets.


\begin{definition}[\textbf{Property $\tU$}] Let $C$ be a set-multilinear
circuit with variable partition $X=\sqcup_{i=1}^d X_i$: Circuit $C$ is
said to have Property $\tU$ if for each monomial $m\in X^d$ all its
proof trees have the same truncated proof tree
type. I.e.\ $\mathcal{\hat{T}}_{C,m}$ is a set of size at most $1$ for
each $m$. Notice that for $m\ne m'$ we can have
$\mathcal{\hat{T}}_{C,m}\ne \mathcal{\hat{T}}_{C,m'}$.
\end{definition}

\begin{theorem}
Any set-multilinear circuit $C$ satisfying Property~$\tU$ w.r.t. the
variable partition $X=\sqcup_{i=1}^n X_i$, where $X_i=\{X_{ij}\mid
1\le j\le n\}$, such that $C$ computes the permanent polynomial
$\PER_n$ requires size $2^{\Omega(n)}$.\footnote{The same lower bound
  result holds for $\DET_n$.}
\end{theorem}

\begin{proof}
Suppose $C$ is a size $s$ set-multilinear circuit satisfying
Property~$\tU$, computing the permanent polynomial $\PER_n$. Let
$G_{n/3}$ denote the set of all product gates $g$ in $C$ such that
$deg(g)>n/3$ and $deg(g_1)\le n/3$ and $deg(g_2)\le n/3$, where $g_1$
and $g_2$ are the gates that are input to $g$. It follows that
$n/3<deg(g)\le 2n/3$. Furthermore, every proof tree of the circuit $C$
has at least one gate from $G_{n/3}$ and at most two gates from
$G_{n/3}$. 

Consequently, by pigeon-hole principle there is an index set $I
\subseteq [n]$, such that truncated proof trees of at least
$n!/s$ many monomials of $\PER_n$ will have a leaf in
$G^I_{n/3}$, where $G^I_{n/3} \subseteq G_{n/3}$ denotes the set
\[
G^I_{n/3}=\{g\in G_{n/3}\mid \text{ index set of }g\text{ is }I\}.
\] 

We will lower bound $|G^I_{n/3}|$. For $g\in G^I_{n/3}$ let $C_g$ be
the subcircuit of $C$ rooted at the gate $g$. Let $\partial_gC$ denote
 the partial derivative of the output gate of $C$ w.r.t.\ gate $g$ as defined in Section~\ref{two}.\\
Notice that circuit for $\partial_gC$ can be  obtained from $C$ as follows:
\begin{itemize}
 \item For all gate $h  \in G^{I}_{n/3}$ such that $h \neq g$, label $0$ for all outgoing edges of $h$.
 \item Replace gate $g$ with constant 1.
 \item For all gate $g' \in G_{n/3} \setminus (G^{I}_{n/3} \cup G^{\overline{I}}_{n/3})$, label 0 for all outgoing edges of $g'$.
\end{itemize}
\noindent
Consider the circuit $C'$ defined as follows:

\begin{eqnarray}\label{eqn-ut}
C'=\sum_{g \in G^{I}_{n/3}}C_g \partial_gC. 
\end{eqnarray}

Since $C_g$ and $\partial_g C$ are both circuits of size at most $s$,
clearly the size of $C'$ is bounded by $s^2$. By choice of $I$ there
are at least $n!/s$ monomials $m$ of $\PER_n$ such that
$\mathcal{\hat{T}}_{C,m}$ has a leaf node labeled $I$. Thus, for every
proof tree in $\mathcal{P}_{C,m}$ there is a gate $g$ in $G^I_{n/3}$.

Crucially, we claim the following.

\begin{claim}
In the polynomial computed by $C'$:
\begin{itemize}
\item The coefficient of every monomial $m$ of $\PER_n$ such that
  $\mathcal{\hat{T}}_{C,m}$ has a leaf node labeled $I$ is $1$.
\item The coefficient of any monomial $m\in X^n$ be any monomial which
  does not occur in $\PER_n$ is $0$.
\end{itemize}
\end{claim}

Both parts of the claim follow from Property~$\tU$. For a monomial
$m\in X^n$, if some proof tree in $\mathcal{P}_{C,m}$ has a gate
$G^I_{n/3}$, then \emph{every} proof tree in $\mathcal{P}_{C,m}$ has a
gate in $G^I_{n/3}$. Hence every proof tree of such a monomial $m$ is
accounted for in the circuit $C'$. Since all proof trees are accounted
for, the net contribution of any such monomial $m$ is the coefficient
of $m$ in $\PER_n$.

We now define a matrix $L$, with respect to index set $I$, as follows:
the rows of $L$ are indexed by monomials $m_1$ of index set $I$. The
columns of $L$ are indexed by gates $t\in G^I_{n/3}$. The entry
$L[m_1,t]$ is the coefficient of the monomial $m_1$ in the subcircuit
computed at gate $t$.

Next, we define a matrix $R$ with respect to index set $I$. The rows
of $R$ are indexed by gates $t\in G^I_{n/3}$, and the columns of $R$
are indexed by monomials $m_2$ such that $m_2$ has index set
$[d]\setminus I$. The entry $R[t,m_2]$ is the coefficient of $m_2$ in
the polynomial $\partial_t C'$.

Let $M_I$ be set of nonzero monomials of $C'$. By the above claim,
$M_I$ is a subset of the nonzero monomial set of $\PER_n$, and the
coefficient of each monomial in $M_I$ is $1$. Furthermore, by choice
of $I$, $|M_I|\ge n!/s$.

Consider monomial $m=m_1m_2$. As argued above, the $(m_1,m_2)^{th}$
entry of matrix $LR$ is the coefficient of $m$ in $\PER_n$ if $m \in
M_I$.

Therefore we have:
\[
   LR[m_1,m_2]= 
\begin{cases}
    1,              & \text{if } m_1m_2 \in M_I,\\
    0,              & \text{otherwise}.\\
\end{cases}
\]

Notice that for any other factorization $m=m'_1m'_2$ of $m$, the entry
$LR[m'_1,m'_2]=0$, because of Property~$\tU$. 

Since $s$ is an upper bound on the ranks of both $L$ and $R$, clearly
$rank(LR)\le s$. We now lower bound the rank of $LR$.

\begin{claim}
The rank of $LR$ is at least $\frac{{n\choose n/3}}{s}$.
\end{claim}

To see the claim, for each subset $S\in {[n]\choose |I|}$ we group
together the rows of matrix $LR$ indexed by monomials $m_1$ of the
form
\[
m_1 = X_{i_1j_1}X_{i2j_2}\ldots X_{i_kj_k},
\]
where $I=\{i_1,i_2,\ldots,i_k\}$ and $S=\{j_1,j_2,\ldots,j_k\}$.

Likewise, corresponding to each subset $T\in{[n]\choose n-|I|}$ we
group together the columns indexed by monomials $m_2$ of the form
\[
m_2 = X_{i_1j_1}X_{i2j_2}\ldots X_{i_\ell j_\ell},
\]
where $[n]\setminus I=\{i_1,i_2,\ldots,i_\ell\}$ and
$T=\{j_1,j_2,\ldots,j_\ell\}$.

We know that that the matrix $LR$ has at least $n!/s$ many $1$'s,
corresponding to the nonzero monomials of $\PER_n$ in $M_I$, and all
other entries of $LR$ are zero.

The matrix $LR$ consists of different $(S,T)$ blocks, corresponding to
subsets $S\in{[n]\choose |I|}$ and $T\in{[n]\choose n-|I|}$. For each
such $S$, only the $(S,[n]\setminus S)$ block has nonzero entries.
All other blocks in the row corresponding to $S$ or the columns
corresponding to $[n]\setminus S$ are zero. Furthermore, we note that
the number of entries in each $(S,[n]\setminus S)$ block is clearly
bounded by $(|I|!)(n-|I|)!$. Therefore, as there are $n!/s$ many $1$'s
in the matrix $LR$, there are at least
\[
{\frac{n!}{s(n-|I|)!|I|!}}={\frac{{n\choose |I|}}{s}}
\]
\emph{nonzero} $(S,[n]\setminus S)$ blocks, each of which contributes
at least $1$ to the rank of the matrix $LR$. Hence,
the rank of the matrix $LR$ is lower bounded by
$\frac{{n\choose |I|}}{s}$.

Putting it together, we obtain
\[
s\ge rank(LR)\ge \frac{{n\choose |I|}}{s}.
\]

Hence $s\ge \sqrt{{n\choose |I|}} \ge \sqrt{{n\choose
    n/3}}=2^{\Omega(n)}$, which completes the proof.
\end{proof}

\subsection{Set-multilinear circuits with few proof tree types}

We now consider set-multilinear circuits with a different restriction
on its proof tree types: Let $C$ be a set-multilinear circuit
computing a degree-$d$ polynomial $f\in\F[X]$ for variable partition
$X=\sqcup_{i=1}^d X_i$ such that \emph{the total number} of proof tree
types of any degree-$d$ monomial in $X^d$ is bounded by a polynomial
in $d$.  Can we prove superpolynomial lower bounds for such circuits?

We are able to show superpolynomial lower bounds in a more restricted
case: when the number of proof trees is bounded by $d^{1/2 -
  \epsilon}$ for some fixed $\epsilon>0$. More precisely, suppose $C$
is a set-multilinear circuit that computes $\PER_n$ and $C$ has at
most $n^{1/2 - \epsilon}$ proof tree types. Then we show that $C$ is
of size $2^{n^{\Omega(1)}}$. We first decompose $C$ into a sum of
$n^{1/2 - \epsilon}$ many set-multilinear formulas $C_i$ such that in
each $C_i$ all proof trees of all monomials have the \emph{same} proof
tree type. Then we convert each $C_i$ into a set-multilinear ABP $A_i$
such that in each layer of this ABP all the nodes are labeled by the
same index set. We can now apply Corollary~\ref{cor:lb-smabp} to the
sum of these $A_i$'s and obtain the claimed lower bound.

\begin{lemma}\label{lem:uni_pt_dr}
 Let $C$ be a set-multilinear circuit of size $s$ computing a
 degree-$d$ polynomial $P\in\F[X]$. If all proof trees in $C$ have the
 same proof tree type $T$, then $C$ can be efficiently transformed
 into a set-multilinear formula $C'$ of size $s^{O(\log n)}$ such that
 in $C'$ too all proof trees have the same proof tree type $T'$, where
 $T'$ depends only on $T$ (and not on the circuit $C$).
\end{lemma} 

\begin{proof}
We prove the lemma by induction on the size of the index set of the
output gate of $C$ (i.e., degree of $P$). At the input gates, where
index set is a singleton set, it clearly holds. Suppose the index set
of the output gate is of size at least $2$.  Let $T_C$ denote the
unique proof tree type for all proof trees in $C$. Each node $v$ of
$T_C$ is labelled by its index set $I_v \subseteq [d]$. As $T_C$ is a
binary tree, there is a vertex $u$ such that $\frac{d}{3} \leq |I_u|
\leq \frac{2d}{3}$. Let $S_u=\{v \in C \mid I_v=I_u\}$. Let
$\hat{C}_v$ denote the set-multilinear circuit obtained from $C$ by
(i) setting to zero all the gates in $S_u \setminus \{v\}$, and (ii)
  replacing the gate $v$ by the constant $1$.  Let $Q_v$ denote the
  polynomial computed at the output gate of $\hat{C}_v$. Notice that
  its index set is $[n] \setminus I_u$. Let $P_v$ denote the
  polynomial computed at a gate $v$ of $C$. Then we can clearly write
\[
 P=\sum_{v \in S_u}P_v Q_v.
\]

Let $C_v$ denote the subcircuit of $C$ with output gate $v$. Note that
\[
\frac{n}{3} \leq deg(P_v),deg(Q_v) \leq \frac{2n}{3}.
\]

Thus, for each $v\in S_u$ both $P_v$ and $Q_v$ are set-multilinear
polynomials computed by set-multilinear circuits ($C_v$ and
$\hat{C}_v$, respectively) of size at most $s$. Furthermore, these
circuits also have the property that all proof trees has the same
proof tree type (otherwise, $C$ would not have the property).

By induction hypothesis, for each $v \in S_u$ we have set-multilinear
formulas $F_v$ and $\hat{F}_v$ such that:

\begin{itemize}
\item $F_v$ and $\hat{F}_v$ compute $P_v$ and $Q_v$, respectively.
\item The size of $F_v$ as well as $\hat{F}_v$ is bounded by
  $s^{O(\log \frac{2n}{3})}$.
\item All proof trees in $F_v$ have a unique proof tree type. All
  proof trees in $\hat{F}_v$ have a unique proof tree type.
\end{itemize}

Furthermore, the circuit $C$ has the following stronger property:
suppose $v$ and $v'$ are two gates with the same index set
$I_v=I_{v'}$. Then the unique proof tree type associated with
subcircuit $C_v$ is the same as the unique proof tree type for
subcircuit $C_{v'}$. Otherwise, the circuit $C$ would not have a
unique proof tree type associated with it. 

Since all the subcircuits $C_v, v\in S_u$ have the same index set and
thus same proof tree type associated to it, it follows by induction
hypothesis that all the formulas $F_v, v\in S_u$ also have the same
unique proof tree type. The same property holds for $\hat{C}_v, v\in
S_u$ and hence $\hat{F}_v, v\in S_u$.

Therefore, each of the product polynomials $P_vQ_v, v\in S_u$,
computed by the formulas $F_v\times\hat{F}_v, v\in S_u$, with a
$\times$ output gate, all have the same proof tree type. Thus, since
$|S_u| \leq s $ the polynomial $P=\sum_{v \in S_u}P_vQ_v$ has a set
multilinear formula $C'$ of size $\leq s(2s^{O(\log \frac{2n}{3})})
\leq s^{O(\log n)}$ and all the proof trees of $C'$ have the same
proof tree type $T'$. Furthermore, it is clear that $T'$ depends only
on $T_C$. This completes the proof of the theorem.
\end{proof}

\begin{lemma}\label{lem:f_to_abp}
 Let $C$ be a set-multilinear formula of size $s$ computing degree $d$
 polynomial $P$, where all proof trees of $C$ have the same proof tree
 type $T$. Then $C$ can be transformed into a set-multilinear ABP such
 that at each layer $i \in [d]$ all gates of layer $i$ is labelled by
 the same index set $I_i$. Furthermore, these index sets
 $I_1,I_2,\ldots,I_d$ depend only on $T$ and not the formula $C$.
\end{lemma}

\begin{proof}
 We show by induction on the size of given formula $C$. 

Suppose the output gate of $C$ is $+$, and $C_1$ and $C_2$ are the two
subformulas. Since $C$ has unique proof tree type $T$, both its
subcircuits $C_1$ and $C_2$ also have the same unique tree type $T$.
By induction hypothesis the two subformulas $C_1$ and $C_2$ of $C$
with same proof tree $T$ can be converted into ABPs $A_1,A_2$
respectively s.t the index sets $I_1,I_2,\ldots, I_d$ is same for both
ABPs. The ``parallel composition'' of these two ABPs yields the
ABP for $C$ with the same index sets.

Suppose output gate of $C$ is $\times$. Since $C$ has unique proof
tree type $T$, both subcircuits $C_1$ and $C_2$ of $C$ has unique
proof tree types, say $T_1,T_2$ respectively. Note that $T_1$ and
$T_2$ are the left and right subtrees of $T$. By induction hypothesis
both $C_1$ and $C_2$ have ABPs with the claimed property. Their
``series composition'' yields the desired ABP with a unique
index set labeling each layer.
\end{proof}

\begin{theorem}\label{thm:sm-bpt-lb}
 Let $C=\sum_{i \in [r]}C_i$ be a set-multilinear circuit, where each
 $C_i$ is set-multilinear, all proof trees of $C_i$ have the same
 proof tree type, and $r=n^{\frac{1}{2}-\epsilon}$, $\epsilon>0$. If
 $C$ computes $\PER_n$ ( or $\DET_n$) then some $C_i$ is of size
 $\Omega(2^{\frac{n^{1/4}}{\log n}})$.
\end{theorem}

\begin{proof}
 The idea is to convert $C$ into a narrow set multilinear ABPs and
 apply the lower bound for narrow set multilinear ABPs
 (Corollary~\ref{cor:lb-smabp}).  By Lemma \ref{lem:uni_pt_dr}, each
 circuit $C_i$ can be converted into a set-multilinear formula $C'_i$
 of unique proof tree type. By Lemma~\ref{lem:f_to_abp} this formula $C'_i$ can
 be transformed into a homogeneous $d$-layer set-multilinear ABP $A_i$
 such that at each layer $i \in [n]$ all the gates in layer $i$ are
 labeled by the same index set $I_i$. Their sum $\sum_{i=1}^r A_i$ is
 a homogeneous $d$-layer set-multilinear ABP A such that at any layer
 $i \in [n]$ the number of different index sets $I_i\subseteq [n]$
 labeling layer $i$ is bounded by $n^{\frac{1}{2}-\epsilon}$. 

The size of ABP $A$ is bounded by $\sum_{i=1}^rs^{O(\log n)} \leq
s^{O(\log n)}$ where $s$ upper bounds the size of each $C_i$.  By
Corollary~\ref{cor:lb-smabp}, any such set-multilinear ABP computing
$\PER_n$ ( or $\DET_n$) requires $\Omega(2^{n^{1/4}})$ size.  Thus,
$s^{O(\log n)} \geq \Omega(2^{n^{1/4}})$, which implies that $s \geq
\Omega(2^{\frac{n^{1/4}}{\log n}})$. This completes the proof of the
theorem.
\end{proof}

\begin{theorem}\label{thm:sm-pt-decom}
 Let $C$ be set multilinear circuit of size $s$ computing the
 polynomial $P\in\F[X]$ of degree $d$ such that the total number of
 proof tree types in $C$ is $r\geq 0$. Then there are set-multilinear
 circuits $C_i, 1\le i\le r$ such that $\sum_{i \in [r]}C_i$ computes
 $P$, each $C_i$ is of size bounded by $s$, and in each $C_i$ all its
 proof trees have the same type.
\end{theorem}

\begin{proof}
 Let the proof tree types of $C$ be $T_1,T_2,\cdots,T_r$. We will
 extract circuit $C_i$ from $C$ corresponding to proof tree type
 $T_i$. In $C_i$, we label 0 for all the outgoing edges of gates $v$
 in $C$ whose index set $I_v\subseteq [n]$ is not equal to any of the
 index sets of proof tree $T_i$. Clearly, the proof trees of $C_i$ are
 precisely all proof trees of proof tree type $T_i$ present in circuit
 $C$ and with the same coefficients. Therefore, $\sum_{i=1}^r C_i$
 computes polynomial $P$. This completes the proof of theorem.
\end{proof}

Combining Theorem~\ref{thm:sm-pt-decom} and ~\ref{thm:sm-bpt-lb}, we
have the following corollary.

\begin{corollary}
 Let $C$ be set multilinear circuit of size $s$ computing the
 polynomial $\PER_n$ (or $\DET_n$). If the total number of distinct
 proof tree types in $C$ is bounded by $c=n^{\frac{1}{2}-\epsilon}$,
 $\epsilon>0$ then $s \geq \Omega(2^{\frac{n^{1/4}}{\log n}})$.
\end{corollary}

\section{Summary and open problems}

In this paper we investigated lower bound questions for certain
set-multilinear arithmetic circuits and ABPs. By imposing a
restriction on the number of set types for set-multilinear ABPs, or by
restricting the number of proof trees in set-multilinear circuits, we
could prove nontrivial lower bounds for the Permanent. We also showed
a separation between set-multilinear circuits and interval multilinear
circuits, assuming the SOS conjecture. 

Some interesting open questions arise from our work: can we show lower
bounds for $f(n)$-narrow set-multilinear ABPs for $f(n)=O(n)$? Another
question is proving lower bounds for set-multilinear circuits with
polynomially (or even $O(n)$) many proof trees computing $\PER_n$. \\

\noindent\textbf{Acknowledgment.}~~We thank Joydeep Mukherjee for
suggesting that Theorem~\ref{thm:mabd} might be useful in the proof of Lemma~\ref{whp-thm}.

\bibliographystyle{amsalpha}

\bibliography{references}

\end{document}